\documentclass[letterpaper,11pt]{article}
\usepackage{verbatim}
\usepackage{times, amssymb, amsthm, amsmath, appendix, verbatim, cancel, graphicx, wrapfig}
\usepackage{textcomp}

\DeclareMathOperator{\Tr}{Tr}

\begin{document}

\title{Generation of Universal Linear Optics by Any Beamsplitter}
\author{Adam Bouland\thanks{%
MIT. \ email: adam@csail.mit.edu.} \and Scott Aaronson\thanks{%
MIT. \ email: aaronson@csail.mit.edu.}}
\date{}
\maketitle

\begin{abstract}
In 1994, Reck \emph{et al}.\ showed how to realize any unitary transformation on a
single photon using a product of beamsplitters and phaseshifters. \ Here we
show that any \textit{single} beamsplitter that nontrivially mixes two
modes, also densely generates the set of unitary transformations (or
orthogonal transformations, in the real case) on the single-photon subspace
with $m\geq 3$ modes. \ (We prove the same result for any two-mode real
optical gate, and for any two-mode optical gate combined with a generic
phaseshifter.) \ Experimentally, this means that one does not need tunable
beamsplitters or phaseshifters for universality: any nontrivial beamsplitter
is universal for linear optics. \ Theoretically, it means that one cannot
produce ``intermediate" models of linear optical computation (analogous to
the Clifford group for qubits) by restricting the allowed beamsplitters and
phaseshifters: there is a dichotomy; one either gets a trivial set or else a
universal set. \ No similar classification theorem for gates acting on
qubits is currently known. \ We leave open the problem of classifying
optical gates that act on three or more modes.
\end{abstract}

\newtheorem{theorem}{Theorem}[section] %
\newtheorem{corollary}[theorem]{Corollary} \newtheorem{lemma}[theorem]{Lemma}
\newtheorem{prop}[theorem]{Proposition} %
\newtheorem{conj}[theorem]{Conjecture} %
\newtheorem{definition}[theorem]{Definition} %
\newtheorem{ex}[theorem]{Example} \newtheorem{rmk}[theorem]{Remark} %
\newtheorem{alg}[theorem]{Algorithm} \newtheorem{lem}[theorem]{Lemma} %
\newtheorem{claim}[theorem]{Claim} \newtheorem{cor}[theorem]{Corollary} %
\newtheorem{defn}{Definition}



\section{Introduction}

Universal quantum computers have proved difficult to build. \ As one
response, researchers have proposed limited models of quantum computation,
which might be easier to realize. \ Three examples are the one clean qubit
model of Knill and Laflamme \cite{Knill1998}, the commuting Hamiltonians
model of Bremner, Jozsa, and Shepherd \cite{BJS}, and the boson sampling
model of Aaronson and Arkhipov \cite{Aaronson2010}. \ None of these models
are known or believed to be capable of universal quantum computation (or,
depending on modeling details, even universal \textit{classical}
computation). \ But all of them can perform certain estimation or sampling
tasks for which no polynomial-time classical algorithm is known.

One obvious way to define a limited model of quantum computation is to
restrict the set of allowed gates. \ However, almost every gate set is
universal \cite{Lloyd1995}, and so are most ``natural" gate sets. \ For
example, Controlled-NOT together with any real one-qubit gate that does not
square to the identity is universal \cite{Shi2002}. \ As a result, very few
nontrivial examples of non-universal gate sets are known. \ All known
non-universal gate sets on $O(1)$ qubits, such as the Clifford group \cite%
{Gottesman2008}, are efficiently classically simulable, if the input and
measurement outcomes both belong to an appropriately chosen qubit basis%
\footnote{%
But not necessarily otherwise! \ For instance, suppose that a nonuniversal
gate set $G$ is efficiently simulable if inputs and outputs are in the
computational basis. \ Now conjugate $G$ by a change of qubit basis to
obtain a gate set $G^{\prime }$. \ Clearly $G^{\prime }$ is efficiently
classically simulable in the new qubit basis. \ However, it is unclear how
to simulate the gates $G^{\prime }$ if inputs and outputs are in the
computational basis. \ Along these lines, there is evidence that Clifford
gates \cite{Jozsa2013}, permutation gates \cite{Jordan}, and even diagonal
gates \cite{BJS} can be hard to simulate in arbitrary bases.}.\ As a result,
it is tempting to conjecture that there does not \textit{exist} such an
intermediate gate set: or more precisely, that any gate set on $O(1)$ qubits
is either efficiently classical simulable (with appropriate input and output
states), or else universal for quantum computing. \ Strikingly, this
dichotomy conjecture remains open even for the special case of 1- and
2-qubit gates! \ We regard proving or disproving the conjecture as an
important open problem for quantum computing theory.

In this paper, we prove a related conjecture in the quantum linear optics
model. \ In quantum optics, the Hilbert space is not built up as a tensor
product of qubits; instead it's built up as a direct \textit{sum} of optical
modes. \ An optical \textit{gate} is then just a unitary transformation that
acts nontrivially on $O(1)$ of the modes, and as the identity on the rest. \
Whenever we have a $k$-mode gate, we assume that we can apply it to any
subset of $k$ modes (in any order), as often as desired. \ The most common
optical gates considered are \textit{beamsplitters}, which act on two modes
and correspond to a $2 \times 2$ unitary matrix with determinant $-1$;%
\footnote{%
Some references use a different convention and assume that beamsplitters
have determinant $+1$ \cite{Nielson2000}. \ Note that these two conventions
are equivalent if one assumes that one can permute modes, i.e.\ apply the
matrix $\bigl(%
\begin{smallmatrix}
0 & 1 \\
1 & 0%
\end{smallmatrix}
\bigr)$ which has determinant $-1$.} and \textit{phaseshifters}, which act
on one mode and simply apply a phase $e^{i\theta}$. \ Note that any unitary
transformation acting on the one-photon Hilbert space automatically gets
``lifted," by homomorphism, to a unitary transformation acting on the
Hilbert space of $n$ photons. \ Furthermore, every element of the \emph{$n$%
-photon linear-optical group}---that is, every $n$-photon unitary
transformation achievable using linear optics---arises in this way (see \cite{Aaronson2010}, Sec. III for details). \ Of course, if $n \ge 2$, then
there are also $n$-photon unitaries that cannot be achieved
linear-optically: that is, the $n$-photon linear-optical group is a proper
subgroup of the full unitary group on the $n$-photon Hilbert space.

We call a set of optical gates $S$ \textit{universal} on $m$ modes if it
generates a dense subset of either $SU(m)$ (in the complex case) or $SO(m)$
(in the real case). \ To clarify, if $S$ is universal, this does \emph{not}
mean that linear optics with $S$ is universal for quantum computing! \ It
only means that $S$ densely generates the one-photon linear-optical
group---or equivalently, the $n$-photon linear-optical group for any value
of $n$. \ The latter kind of universality is certainly \emph{relevant} for
quantum computation: first, it already suffices for the boson sampling
proposal of Aaronson and Arkhipov \cite{Aaronson2010}; and second, if the
single resource of adaptive measurements is added, then universal linear
optics becomes enough for universal quantum computation, by the famous
result of Knill, Laflamme, and Milburn (KLM) \cite{Knill2001}. \ On the
other hand, if we wanted to map a $k$-qubit Hilbert space \emph{directly}
onto an $m$-mode linear-optical Hilbert space, then as observed by Cerf,
Adami and Kwiat \cite{Cerf1998}, we would need $m \ge 2^k$ just for
dimension-counting reasons.

Previously, Reck \emph{et al}.\ \cite{Reck1994} showed that the set of \emph{all}
phaseshifters and all beamsplitters is universal for linear optics, on any
number of modes. \ Therefore it is natural to ask: is there \emph{any} $S$
set of beamsplitters and phaseshifters that generates a nontrivial set of
linear-optical transformations, yet that still falls short of generating
\textit{all} of them? \ Here by \textquotedblleft
nontrivial,\textquotedblright\ we simply mean that $S$ does \emph{something}
more than permuting the modes around or adding phases to them.

If such a set $S$ existed, we could then ask the \emph{further} question of
whether the $n$-photon subgroup generated by $S$ was

\begin{enumerate}
\item[(a)] efficiently simulable using a classical computer, despite being
nontrivial (much like the Clifford group for qubits),

\item[(b)] already sufficient for applications such as boson sampling and KLM,
despite not being the full $n$-photon linear-optical group, or

\item[(c)] of \textquotedblleft intermediate\textquotedblright\ status,
neither sufficient for boson sampling and KLM nor efficiently simulable
classically.
\end{enumerate}

The implications for our dichotomy conjecture would of course depend on the
answer to that further question.

In this paper, however, we show that the further question never even arises,
since \textit{no such set }$S$\textit{\ exists}.\ \ Indeed, any beamsplitter
that acts nontrivially on two modes is universal on three or more modes. \
What makes this result surprising is that it holds \textit{even if the
beamsplitter angles are all rational multiples of }$\pi $\textit{.} \ A
priori, one might guess that by restricting the beamsplitter angles to (say)
$\pi /4$, one could produce a linear-optical analogue of the Clifford group;
but our result shows that one cannot.

Our proof uses standard representation theory and the classification of
closed subgroups of $SU(3)$ \cite{Fairbairn1964,Hanany1999, GrimusLudl2012}. \ From an
experimental perspective, our result shows that any complex nontrivial
beamsplitter suffices to create any desired optical network. \ From a
computational complexity perspective, it implies a dichotomy theorem for
optical gate sets: any set of beamsplitters or phaseshifters generates a set
of operations that is either \textit{trivially} classically simulable (even
on $n$-photon input states), or else universal for quantum linear optics. \
In particular, any nontrivial beamsplitter can be used to perform
boson sampling; there is no way to define an \textquotedblleft
intermediate\textquotedblright\ model of boson sampling\footnote{%
Here by \textquotedblleft intermediate,\textquotedblright\ we mean
computationally intermediate between classical computation and universal
boson sampling.} by restricting the allowed beamsplitters and phaseshifters.

Note that our result holds only for beamsplitters, i.e., optical gates that
act on two modes and have determinant $-1$. \ We leave as an open problem
whether our result can be extended to arbitrary two-mode gates, or to gates
that act on three or more modes.

Our work is the first that we know of to explore limiting the power of
quantum linear optics by limiting the gate set. \ Previous work has
considered varying the available input states and measurements. \ For
example, as mentioned earlier, Knill, Laflamme, and Milburn \cite{Knill2001}
showed that linear optics with adaptive measurements is universal for
quantum computation. \ Restricting to nonadaptive measurements seems to
reduce the computational power of linear optics, but Aaronson and Arkhipov
\cite{Aaronson2010} gave evidence that the resulting model is still
impossible to simulate efficiently using a classical computer. \ If Gaussian
states are used as inputs and measurements are taken in the Gaussian basis
only, then the model is efficiently simulable classically \cite{Bartlett2003}%
; but with Gaussian-state inputs and \textit{photon-number} measurements,
there is recent evidence for computational hardness.\footnote{%
See http://www.scottaaronson.com/blog/?p=1579}

We hope that this work will serve as a first step toward proving the
dichotomy conjecture for \emph{qubit}-based quantum circuits (i.e., the
conjecture that every set of gates is either universal for quantum
computation or else efficiently classically simulable). \ The tensor product
structure of qubits gives rise to a much more complicated problem than the
direct sum structure of linear optics. \ For that reason, one might expect
the linear-optical \textquotedblleft model case\textquotedblright\ to be
easier to tackle first, and the present work confirms that expectation.

\section{Background and Our Results}

In a linear optical system with $m$ modes, the state of a photon is
described by a vector $\mbox{$ | \psi \rangle $}$ in an $m$-dimensional
Hilbert space. \ The basis states of the system are represented by strings $%
\mbox{$ | s_1, s_2 \ldots s_m \rangle $}$ where $s_i\in\{0,1\}$ denotes the
number of photons in the $i^{\mathrm{th}}$ mode, and $\Sigma_{j=1}^{m} s_j$
is the total number of photons (in this case, one). \ For example a one%
-photon, three-mode system has basis states $\mbox{$ | 100 \rangle $},\mbox{$
| 010 \rangle $}$ and $\mbox{$ | 001 \rangle $}$.

A $k$-local gate $g$ is a $k \times k$ unitary matrix which acts on $k$
modes at a time while acting in direct sum with the identity on the
remaining $m-k$ modes. \ A \emph{beamsplitter} $b$ is a two-local gate with
determinant $-1$. \ Therefore any beamsplitter has the form $b=\left(%
\begin{smallmatrix}
\alpha & \beta^* \\
\beta & -\alpha^*%
\end{smallmatrix}
\right)$ where $|\alpha|^2 + |\beta|^2 = 1$. \ Let $b_{ij}$ denote the
matrix action of applying the beamsplitter to modes $i$ and $j$ of a
one-photon system. \ For example, if $m=3$, we have that
\begin{equation*}
\begin{matrix}
b_{12}= \left(%
\begin{matrix}
\alpha & \beta^* & 0 \\
\beta & -\alpha^* & 0 \\
0 & 0 & 1%
\end{matrix}
\right) &  & b_{31}= \left(%
\begin{matrix}
-\alpha^* & 0 & \beta \\
0 & 1 & 0 \\
\beta^* & 0 & \alpha%
\end{matrix}
\right)%
\end{matrix}%
\end{equation*}
when written in the computational basis. \ A beamsplitter is called \emph{%
nontrivial} if $|\alpha|\neq 0$ and $|\beta|\neq 0$, i.e.\ if the
beamsplitter mixes modes.

We say that a set $S$ of optical gates \emph{densely generates} a continuous
group $G$ of unitary transformations, if the group $H$ generated by $S$ is a
dense subgroup of $G$ (that is, if $H\leq G$ and $H$ contains arbitrarily
close approximations to every element of $G$). \ Then we call $S$ \emph{%
universal on $m$ modes} if it densely generates $SU(m)$ or $SO(m)$ when
acting on $m$ modes. \ (Due to the irrelevance of global phases, this is
physically equivalent to generating $U(m)$ or $O(m)$ respectively.) \ In
this definition we are assuming that whenever we have a $k$-mode gate in $S$%
, we can apply it to any subset of $k$ modes (in any order), as often as
desired. \ Note that we consider real $SO(m)$ evolutions to be universal as
well; this is because the distinction between real and complex optical
networks is mostly irrelevant\footnote{%
The one case we know about where the real vs. complex distinction might matter
is when using error-correcting codes. \ There, applying all possible
orthogonal transformations to the physical modes or qubits might not suffice to
apply all orthogonal transformations to the encoded modes or qubits. \ This
could conceivably be an issue, for example, in the scheme of Gottesman,
Kitaev, and Preskill \cite{GottesmanKitaevPreskill} for universal quantum
computing with linear optics.} to computational applications of linear
optics, such as the KLM protocol \cite{Knill2001} and boson sampling \cite%
{Aaronson2010}.

A basic result in quantum optics, proved by Reck \emph{et al}.\ \cite{Reck1994},
says that the collection of all beamsplitters and phaseshifters is
universal. \ Specifically, given any target unitary $U$ on $m$ modes, there
exists a sequence of $O(m^{2})$ beamsplitters and phaseshifters whose
product is exactly $U$. \ Reck \emph{et al}.'s proof also shows an analogous result
for real beamsplitters - namely, that any orthogonal matrix $O$ can be
written as the product of $O(m^{2})$ real beamsplitters. \ Furthermore, it
can easily be shown that there exist two beamsplitters $b$, $b^{\prime }$
whose products densely generate $O(2)$. \ Therefore $b$ and $b^{\prime }$
can be used to simulate any real beamsplitter, and hence by Reck \emph{et al}.\
\cite{Reck1994}, the set $\{b,b^{\prime }\}$ is universal for linear optics.

In this paper, we consider the universality of a \emph{single} beamsplitter $%
b$. \ If $b$ is trivial, then on $m$ modes the matrices $b_{ij}$ generates a
subgroup of $P_{m}$, the set of $m\times m$ unitary matrices with all
entries having norm zero or one. \ This is obviously non-universal, and the
state evolutions on any number of photons are trivial to simulate
classically. \ Our main result is that any nontrivial beamsplitter densely
generates either all orthogonal transformations on three modes (in the real
case), or all unitary transformations on three modes (in the complex case). \
From this, it follows easily from Reck \emph{et al}.\ \cite{Reck1994} that such a
beamsplitter is also universal on $m$ modes for any $m\geq 3$.

\begin{theorem}\label{o3_generation}
Let $b$ be any nontrivial beamsplitter. \ Then the set $S=\{b_{12},b_{13}, b_{23}\}$, obtained by applying $b$ to all possible pairs among three photon modes,\footnote{Technically, we could also consider the unitaries $b_{21},b_{31},b_{32}$, obtained by applying $b$ to the same pairs of modes but reversing their order. \ However, this turns out not to give us any advantage.} densely generates either $SO(3)$ (if all entries of $b$ are real) or $SU(3)$ (if any entry of $b$ is non-real).
\end{theorem}

\begin{cor} Any nontrivial beamsplitter is universal on $m\geq 3$ modes. \label{cor:universality}
\end{cor}
\begin{proof}

By Theorem \ref{o3_generation}, the set $S=\{b_{12},b_{13}, b_{23}\}$ densely generates all orthogonal matrices with determinant $1$. \ But since $b$ has determinant $-1$, we know that $S$ must generate all orthogonal matrices with determinant $-1$ as well.\footnote{Indeed any orthogonal $O$ with determinant $-1$ can be written as $O=b_{12}^{-1}O' = b_{12}O'$ for some $O'$ of determinant $1$.} \ Therefore, $S$ densely generates the action of any real beamsplitter $b'$ acting on two out of three modes. \ So by Reck et al \cite{Reck1994}, $S$ also densely generates all orthogonal matrices on $m$ modes for $m\geq 3$.
\end{proof}

Note that, although our proof of universality on three modes is
nonconstructive, by the Solovay-Kitaev Theorem \cite{Dawson2005}, there is
an efficient algorithm that, given any target unitary $U$, finds a sequence
of $b$'s approximating $U$ up to error $\varepsilon$ in $O\left(\log^{3.97}(%
\frac{1}{\varepsilon})\right)$ time. \ Thus, our universality result also
implies an efficient algorithm to construct any target unitary using
beamsplitters in the same manner as Reck \emph{et al}.\ \cite{Reck1994}.

We now proceed to a proof of Theorem \ref{o3_generation}.

\section{Proof of Main Theorem}

We first consider applying a fixed beamsplitter

\begin{equation*}
b=\left(%
\begin{matrix}
\alpha & \beta^* \\
\beta & -\alpha^*%
\end{matrix}
\right),
\end{equation*}

\noindent where $\alpha $ and $\beta $ are complex and $|\alpha |^{2}+|\beta |^{2}=1$,
to two modes of a three-mode optical system. \ We take pairwise products of
these beamsplitter actions to generate three special unitary matrices. \ These
three unitaries densely generate some group of matrices $G\leq SU(3)$. \ We
then use the representation theory of subgroups of $SU(3)$ described in the
work of Fairbairn, Fulton \& Klink \ \cite{Fairbairn1964}, Hanany \&
He \cite{Hanany1999}, and Grimus \& Ludl \cite{GrimusLudl2012} to show that the beamsplitter must generate either all $
SO(3)$ matrices (if the beamsplitter is real) or all $SU(3)$ matrices (if
the beamsplitter has a complex entry).

Consider applying our beamsplitter to a three-mode system. \ Let $R_1,R_2,R_3$
be defined as the pairwise products of the beamsplitter actions below:
\begin{equation*}
\begin{matrix}
R_1= b_{12}b_{13}=\left(
\begin{matrix}
\alpha^2 & \beta^* & \alpha\beta^* \\
\alpha\beta & -\alpha^* & |\beta|^2 \\
\beta & 0 & -\alpha^*%
\end{matrix}
\right), & R_2 = b_{23}b_{13}=\left(
\begin{matrix}
\alpha & 0 & \beta^* \\
|\beta|^2 & \alpha & -\alpha^*\beta^* \\
-\alpha^*\beta & \beta & \alpha^{*2}%
\end{matrix}
\right),%
\end{matrix}
\end{equation*}
\begin{equation*}
R_3 = b_{12}b_{23}=\left(
\begin{matrix}
\alpha & \alpha\beta^* & \beta^{*2} \\
\beta & -|\alpha|^2 & -\alpha^*\beta^* \\
0 & \beta & -\alpha^*%
\end{matrix}
\right).
\end{equation*}

Since $R_1, R_2, R_3$ are even products of matrices of determinant $-1$,
they are all elements of $SU(3)$. \ Let $G \le SU(3)$ be the subgroup
densely generated by products of the elements $\{R_1,R_2,R_3\}$ and their inverses\footnote{Since $b_{ij}^{-1}=b_{ij}$, the beamsplitter is capable of generating the inverses of the $R_i$ as well.}. \ Let $G_M$ be the set of
matrices representing $G$ under this construction. \ First we will show that
these matrices $G_M$ form an irreducible representation of $G$.

\begin{claim}\label{clm:irrep}
The set $\{R_1, R_2, R_3\}$ generates an irreducible three-dimensional representation of $G$.
\end{claim}
\begin{proof}
Suppose that some matrix
\[
U=\left(\begin{matrix}
A & D & G\\
B & E & H\\
C & F & I \\
\end{matrix}\right)
\]
commutes with $R_1$, $R_2$, and $R_3$. \ Then we claim that $U$ is a constant multiple of the identity, i.e.\ $A=E=I$ and $D=G=H=B=C=F=0$.

From the claim, it follows easily that the representation is irreducible. \ Indeed, suppose the representation is reducible, so preserves a non-trivial subspace. \ Since our representation is unitary, this implies that our representation is decomposable, i.e.\ by a change of basis it can be brought into block-diagonal form.\footnote{To see the equivalence of ``reducible'' and ``decomposable'' for unitary representations, it suffices to note that, if a set of unitary matrices always map a subspace $V$ to itself, then they cannot map any vector not in $V$ to a vector in $V$, since this would violate unitarity.} \ In the new basis, the matrix consisting of 1's on the diagonal in the first block, and 2's in the diagonal of the second block, commutes with all elements of $G$, and in particular with $R_1,R_2,R_3$. \ But that matrix is not a multiple of the identity. \ Hence if only multiples of the identity commute with $R_1,R_2,R_3$, the representation must be irreducible.

We now prove the claim. \ First, since $U$ commutes with $R_1$,
\[
\left(\begin{matrix}
A & D & G\\
B & E & H\\
C & F & I \\
\end{matrix}\right)
\left(\begin{matrix}
\alpha^2 & \beta^* & \alpha\beta^* \\
\alpha\beta & -\alpha^* & |\beta|^2 \\
\beta & 0 & -\alpha^*
\end{matrix}\right)
= \left(\begin{matrix}
\alpha^2 & \beta^* & \alpha\beta^* \\
\alpha\beta & -\alpha^* & |\beta|^2 \\
\beta & 0 & -\alpha^*
\end{matrix}
\right)\left(\begin{matrix}
A & D & G\\
B & E & H\\
C & F & I \\
\end{matrix}\right).
\]
This imposes nine equations. \ Below we give the equations coming from the (1,1), (1,2), (2,2), (2,3), and (3,2) entries of the above matrices respectively.
\begin{align}
(D\alpha + G)\beta &= (C\alpha + B)\beta^* \label{eq:1}, \\
(A-E-F\alpha)\beta^* &= D(\alpha^2 +\alpha^*)  \label{eq:2},\\
B\beta^*&=D\alpha\beta+F\beta\beta^*  \label{eq:5},\\
B\alpha\beta^*+E\beta\beta^*-H\alpha^* &= G\alpha\beta -H\alpha^* +I\beta\beta^*  \label{eq:6},\\
C\beta^*&=D\beta  \label{eq:8}.
\end{align}

Note that Eqs. (\ref{eq:8}) and (\ref{eq:1}) imply that
\begin{equation}\label{eq:10}G\beta = B\beta^* .\end{equation}
So by Eq. (\ref{eq:6}) we have
\begin{equation} \label{eq:11}E\beta\beta^* = I\beta\beta^* .\end{equation}
So since $0<|\beta|<1$, we have $I=E$.

In total so far we have $I=E$, $G\beta = B\beta^*$ and $C\beta^*= D\beta$.

Next, since $U$ commutes with $R_2$,
\[
\left(\begin{matrix}
A & D & G\\
B & E & H\\
C & F & E \\
\end{matrix}\right)
\left(
\begin{matrix}
\alpha & 0 & \beta^* \\
|\beta|^2 & \alpha & -\alpha^*\beta^* \\
-\alpha^*\beta & \beta & \alpha^{*2}
\end{matrix}
\right)
= \left(
\begin{matrix}
\alpha & 0 & \beta^* \\
|\beta|^2 & \alpha & -\alpha^*\beta^* \\
-\alpha^*\beta & \beta & \alpha^{*2}
\end{matrix}
\right)\left(\begin{matrix}
A & D & G\\
B & E & H\\
C & F & E \\
\end{matrix}\right).
\]
This imposes another nine equations. \ Here are the equations from the (1,1), (2,1) and (2,2) entries respectively, which we have simplified using $I=E$, $G\beta = B\beta^*$ and $C\beta^*= D\beta$:
\begin{align}
D\beta &= D\beta\beta^* -G\alpha^*\beta   \label{eq:18},\\
E\beta\beta^*-H\alpha^*\beta &= A\beta\beta^* -C\alpha^*\beta^*  \label{eq:21},\\
H\beta &= D\beta\beta^* -F\alpha^*\beta^*  \label{eq:22}.
\end{align}

Note that Eqs. (\ref{eq:18}) and (\ref{eq:22}), combined with the fact that $G\beta = B\beta^*$, imply that $D\beta=H\beta$, and hence $D=H$.

Plugging this in to Eq. (\ref{eq:21}), we see that $E\beta\beta^*-D\alpha^*\beta = A\beta\beta^* -C\alpha^*\beta^*$. Using $C\beta^*= D\beta$ these last two terms cancel, so $E\beta\beta^* = A\beta\beta^*$, and hence $E=A$. \ So overall we have established that $A=E=I$, $D=H$, $B=F$, $G\beta = B\beta^*$ and $C\beta^*= D\beta$.

Now suppose $B=0$. \ Then we have from above that $B=F=G=0$. \ By Eq. (\ref{eq:18}) we also have $D\beta = D\beta\beta^* \Rightarrow D=0$ since $0<|\beta|<1$. \ Hence we have $C=0$ as well by the fact that $C\beta^*= D\beta$. \ Therefore $U$ is a multiple of the identity, as desired.

So it suffices to prove that $B=0$. \ Suppose $B\neq0$; then we will derive a contradiction.

Since $U$ commutes with $R_3$,
\[
\left(\begin{matrix}
A & D & G\\
B & A & D\\
C & B & A \\
\end{matrix}\right)
\left(
\begin{matrix}
\alpha & \alpha\beta^* & \beta^{*2} \\
\beta & -|\alpha|^2 & -\alpha^*\beta^* \\
0 & \beta & -\alpha^*
\end{matrix}
\right)
= \left(
\begin{matrix}
\alpha & \alpha\beta^* & \beta^{*2} \\
\beta & -|\alpha|^2 & -\alpha^*\beta^* \\
0 & \beta & -\alpha^*
\end{matrix}
\right)\left(\begin{matrix}
A & D & G\\
B & A & D\\
C & B & A \\
\end{matrix}\right).
\]
This imposes yet another nine equations, but we will only need the one coming from the (2,2) entry of the above matrices to complete the proof:
\begin{align}
B\alpha\beta^* &= -B\alpha^*\beta^* .\label{eq:27}
\end{align}

Since $B\neq 0$, Eq. (\ref{eq:27}) implies that $\alpha=-\alpha^*$, i.e.\ $\alpha$ is pure imaginary. \ Furthermore, since $G\beta = B\beta^*$, we have $G\neq0$ as well. \ Using this, we can write out Eqs. (\ref{eq:2}) and (\ref{eq:5}) as follows:
\begin{alignat}{6}
(-B\alpha)\beta^*  &= D(\alpha^2 -\alpha) &\quad \Rightarrow &\quad  G\beta &= D(1-\alpha)  \label{eq:28},\\
B\beta^*  &= D\alpha\beta+F\beta\beta^* &\quad \Rightarrow   &\quad G &= D\alpha + G\beta \label{eq:29}.
\end{alignat}

Summing these equations, we see that $G=D$. \ Plugging back into Eq. (\ref{eq:29}), we see that $\beta = 1 -\alpha$. \ Since $\alpha$ is pure imaginary this contradicts $|\alpha|^2+|\beta|^2=1$.

To summarize, if $U$ commutes with all elements of $G$, then $U$ is a multiple of the identity. \ This proves the claim and hence the theorem.
\end{proof}

We have learned that the set $G_M$ forms a three-dimensional irreducible
representation of $G$. \ We now leverage this fact, along with the
classification of finite subgroups of $SU(3)$, to show that $G$ is not
finite.

\begin{claim} \label{inf} $G$ is infinite.
\end{claim}
\begin{proof}
By Claim \ref{clm:irrep}, if $G$ is finite then $\{R_1,R_2,R_3\}$ generates an irreducible representation of $G$. \ The finite subgroups of $SU(3)$ consist of the finite subgroups of $SU(2)$, $12$ exceptional finite subgroups, and two infinite families of ``dihedral-like" groups, whose irreducible representations are classified in \cite{Fairbairn1964,Hanany1999,GrimusLudl2012}. \ Our proof proceeds by simply enumerating the possible finite groups that $G$ could be, and showing that $\{R_1,R_2,R_3\}$ cannot generate an irreducible representation of any of them.

First we eliminate the possibility that $G$ is an exceptional finite subgroup of $SU(3)$. \ Of the $12$ exceptional subgroups, only eight of them have three-dimensional irreps: they are labeled $\Sigma(60)$, $\Sigma(60) \times \mathbb{Z}_3$, $\Sigma(168)$, $\Sigma(168) \times \mathbb{Z}_3$, $\Sigma(216)$, $\Sigma(36\times 3)$, $\Sigma(216\times 3)$, and $\Sigma(360\times 3)$. \ So by Claim \ref{clm:irrep}, if $G$ is finite and exceptional, then it is one of these eight groups.

The character tables of these groups are provided in \cite{Fairbairn1964} and \cite{Hanany1999}. \ Recall that the character of an element of a representation is the trace of its representative matrix. \ The traces of the matrices $R_1,R_2,R_3$, denoted $T_1,T_2,T_3$, are given by
\begin{align}
T_1 &=\alpha^2 - 2\alpha^*\label{eq:tr1}\\
T_2 &=(\alpha^*)^2 +2\alpha\label{eq:tr2}\\
T_3 &= -|\alpha|^2 + \alpha-\alpha^* = -|\alpha|^2 +2\mathrm{Im}(\alpha)\label{eq:tr3}
\end{align}
We will show that these cannot be the characters of the elements of a three-dimensional irrep of $\Sigma(60)$, $\Sigma(60) \times \mathbb{Z}_3$, $\Sigma(168)$, $\Sigma(168) \times \mathbb{Z}_3$, $\Sigma(216)$, $\Sigma(36\times 3)$, $\Sigma(216\times 3)$, or $\Sigma(360\times 3)$. 

There are two three-dimensional irreps of $\Sigma(60)$ up to conjugation \cite{Fairbairn1964}. \ The characters of their elements all lie in the set $\left\{ 0, -1, 3, \frac{1+\sqrt{5}}{2}, \frac{1-\sqrt{5}}{2} \right\}$. \ Note that $0<|\alpha|^2<1$, which means that $T_3$ cannot be in this set unless $T_3=\frac{1-\sqrt{5}}{2}$ and $\mathrm{Im}(\alpha)=0$. \ But then this implies $\alpha=\pm\sqrt{\frac{\sqrt{5}-1}{2}}$. \ Plugging this into $T_1$ and $T_2$, we see they are not in the set of allowed values. \ Hence $G$ is not $\Sigma(60)$.

The characters of the three-dimensional irreps of $\Sigma(60)\times\mathbb{Z}_3$ are identical to those of $\Sigma(60)$, but with the additional possibility that they can be multiplied by $e^{\frac{2\pi i}{3}}$ or $e^{\frac{4\pi i}{3}}$. The same argument as above shows that in order for $T_3$ to be in the set of allowed characters, we must have $\alpha=\pm\sqrt{\frac{\sqrt{5}-1}{2}}$ or $\alpha=\frac{1+\sqrt{5}}{8}\left(\pm 1 \pm \sqrt{3} i \right)$. Plugging these into $T_1$, we see the possible values of $T_1$ are not in the set of allowed characters, hence $G$ is not $\Sigma(60)\times \mathbb{Z}_3$.

There are two three-dimensional irreps of $\Sigma(168)$ up to conjugation \cite{Fairbairn1964}. \ The characters of their elements all lie in the set $S=\left\{0,\pm1,3, \frac{1}{2}(-1\pm i\sqrt{7})\right\}$. \ Since $0<|\alpha|^2<1$, if $T_3$ is in this set it must have value $\frac{1}{2}(-1\pm i\sqrt{7})$. \ Therefore we must have $\alpha = \pm\frac{1}{4} \pm \frac{\sqrt{7}}{4}i$. \ This implies that $\alpha^2 = \frac{-3}{8} \pm \frac{\sqrt{7}}{8} i$ and $2\alpha^* = \pm \frac{1}{2} \pm \frac{\sqrt{7}}{2}i$. \ Regardless of the signs chosen, this means that $T_1$ is not in the set $S$ of allowed values. \ Hence $G$ is not $\Sigma(168)$.

The characters of the three-dimensional irreps of $\Sigma(168)\times\mathbb{Z}_3$ are identical to those of $\Sigma(168)$, but with the additional possibility that they can be multiplied by $e^{\frac{2\pi i}{3}}$ or $e^{\frac{4\pi i}{3}}$. The same argument as above shows that in order for $T_3$ to be in the set of allowed characters, we must have $\alpha=\pm \frac{1}{4} \pm \frac{\sqrt{7}}{4} i$, $\alpha=\frac{1}{4}(\pm\sqrt{5} \pm i \sqrt{3})$, or $\alpha=\pm\frac{1}{4\sqrt{2}}\sqrt{7\sqrt{21}-13} \pm \frac{\sqrt{3}+\sqrt{7}}{8} i$. Plugging these into $T_1$, we see that the possible values of $T_1$ are not in the set of allowed characters, hence $G$ is not $\Sigma(168)\times \mathbb{Z}_3$.

There is one three-dimensional irrep of $\Sigma(216)$ up to conjugation \cite{Fairbairn1964}. \ The characters of its elements all lie in the set $\{0,-1,3\}$. \ Since $T_3$ cannot be in this set, $G$ is not $\Sigma(216)$.

There are eight three-dimensional irreps of $\Sigma(36\times 3)$ up to conjugation \cite{Hanany1999}. \ The characters of their elements all lie in the set $S=\{0, \pm1, \pm e_3, \pm e_3^2, \pm e_4, \pm e_{12}^7, \pm e_{12}^{11} \pm3, \pm3e_3, \pm 3e_3^2\}$ where $e_n = e^{\frac{2\pi i}{n}}$. \ Since $\mathrm{Re}(T_3) = -|\alpha|^2$ and $0<|\alpha|^2<1$, if $T_3 \in S$ then we must have $T_3\in\{\pm e_3, \pm e_3^2, \pm e_{12}^7, \pm e_{12}^{11} \}$. \ Solving for $\alpha$ gives us $\alpha \in \left\{\frac{\pm \sqrt{5}\pm \sqrt{3}i}{4}, \frac{\pm\sqrt{8\sqrt{3}-1}\pm i}{4} \right\}$. \ A straightforward evaluation of possible values of $T_1$ shows $T_1 \notin S$. \ So $T_1$ and $T_3$ cannot be characters of these irreps, and hence $G$ is not $\Sigma(36\times 3)$.

There are seven three-dimensional irreps of $\Sigma(216\times 3)$ up to conjugation \cite{Hanany1999}. \ The characters of their elements all lie in the set \[S=\{0, \pm1, 3, \pm e_3,\pm e_3^2, -e_9^2, -e_9^4, -e_9^5, -e_9^7, \pm e_9^2 +e_9^5, 2e_9^2 +e_9^5, -e_9^2 -2e_9^5, e_9^4+e_9^7, e_9^4 + 2e_9^7, -2e_9^4-e_9^7 \}.\]
If $T_3 \in S$, then for each case we can solve for $\alpha$ and hence $T_1$. \ As above, a straightforward calculation shows that for no $T_3 \in S$ do we have $T_1 \in S$. \ Hence $G$ is not $\Sigma(216\times 3)$.

There are four three-dimensional irreps of $\Sigma(360\times 3)$ up to conjugation \cite{Hanany1999}. \ The characters of their elements all lie in the set \[S=\{0, \pm1, \pm e_3,\pm e_3^2, 3e_3,3e_3^2, -e_5 -e_5^4, -e_5^2 -e_5^3,-e_{15} -e_{15}^4 ,-e_{15}^7 -e_{15}^{13} ,-e_{15}^{11} -e_{15}^{14} ,-e_{15}^2 -e_{15}^8 \}.\]
Again a straightforward calculation shows that for no $T_3 \in S$ do we have $T_1 \in S$. \ Hence $G$ is not $\Sigma(360\times 3)$.

We have therefore shown that $G_M$ is not an irrep of an exceptional finite subgroup of $SU(3)$.

Next we will show that $G_M$ is not in one of the two infinite families of ``dihedral-like" finite subgroups of $SU(3)$, known as the C-series and the D-series groups. \ The most well-known members of these series are $\Delta(3n^2)$ and $\Delta(6n^2)$, labeled by $n \in \mathbb{N}$, which consist of all 3 by 3 even permutation matrices (for $\Delta(3n^2)$) or all 3 by 3 permutation matrices (for $\Delta(6n^2)$) whose entries are replaced by $n$th roots of unity. \ In early works describing subgroups of $SU(3)$, such as Fairbairn, Fulton, and Klink \cite{Fairbairn1964}, only  $\Delta(3n^2)$ and $\Delta(6n^2)$ appear as elements of these series. \ However in 2011, Ludl \cite{Ludl2011} pointed out that there exist nontrivial subgroups of $\Delta(3n^2)$ and $\Delta(6n^2)$ which are missing from these references. \ Fortunately these groups have now been fully classified \cite{GrimusLudl2014}, and sufficient constraints have been placed on their representations \cite{GrimusLudl2012} that we can eliminate the possibility that $G$ is an irrep of any C or D-series group.

In the following, we first eliminate the possibility that $G_M$ is an irrep of $\Delta(3n^2)$ or $\Delta(6n^2)$ following the work of Fairbairn, Fulton, and Klink \cite{Fairbairn1964}. \ Afterwards we show that these arguments suffice to prove $G_M$ is not an irrep of any of the C-series or D-series groups, using the work of Grimus and Ludl \cite{GrimusLudl2012}.

The three-dimensional irreps of $\Delta(3n^2)$ are labeled by integers $m_1,m_2 \in\{0, \ldots, n-1\}$, and have conjugacy classes labeled by types $A,C,E$ and numbers $p,q \in\{0, \ldots, n-1\}$. \ The respective characters are either $0$ for conjugacy classes $C(p,q)$ and $E(p,q)$ or
\begin{equation}\label{eq:3n2}e^{\frac{2\pi i}{n}\left(m_1 p +m_2 q\right)}+e^{\frac{2\pi i}{n}\left(m_1 q -m_2(p+ q)\right)}+e^{\frac{2\pi i}{n}\left(-m_1( p+q) +m_2 p\right)}\end{equation}
for conjugacy class $A(p,q)$.

Assume that $G_M$ is an irrep of $\Delta(3n^2)$ for some $n$---we will derive a contradiction shortly. \ Then the trace of each $R_i$ must be zero (if $R_i$ is a representative of type $C$ or $E$) or of the form of Eq. (\ref{eq:3n2}) (if $R_i$ is a representative of type $A$). \ However, we can show that none of the traces $T_i$ can be $0$ because our beamsplitter is nontrivial. \ Indeed $T_3$ cannot be zero as $0<|\alpha|^2<1$. \ We know that in order for $T_1$ to be zero, we need $\alpha^2 = 2\alpha^*$, which implies $|\alpha|=2$ which is not possible, and likewise with $T_2$. \ Hence each $T_i$ must have the form of Eq. (\ref{eq:3n2}), which implies each $R_i$ is in conjugacy class $A(p_i,q_i)$ for some choice of $p_i,q_i$. \ However, looking at the multiplication table for this group provided in \cite[Table VIII]{Fairbairn1964}, we have that $A(p,q)A(p',q')=A\left(p+q \mod{n},p'+q' \mod{n}\right)$. \ Hence the $T_i$'s cannot possibly generate all of $\Delta(3n^2)$ for any $n$, since they cannot generate elements in the conjugacy classes $C(p,q)$ or $E(p,q)$. \ This contradicts our assumption that the $R_i$'s generate an irrep of $\Delta(3n^2)$. \ Therefore $G_M$ is not an irrep of $\Delta(3n^2)$ for any $n$.

We now extend this argument to eliminate the possibility that $G_M$ is an irrep of any of the C-series groups. \ In Appendix E of \cite{GrimusLudl2012}, Grimus and Ludl show that for any three-dimensional irrep of a C-series group, there exists a basis (and an ordering of that basis) in which all elements of the $A$ conjugacy classes are represented by diagonal matrices and $E(0,0)$ is represented by 
\begin{equation} \label{eq:Ematrix}
\left(\begin{matrix}0 & 1 & 0 \\ 0 & 0 &1 \\ 1 & 0 & 0 \end{matrix}\right).
\end{equation}
From this it can be easily shown that in any three-dimensional irrep of a C-series group, all elements of types $C$ and $E$ are represented by traceless matrices.\footnote{To see this, note that by the group multiplication table in [13] Table VIII, we have that $A(p,q)=E(p,q)E(0,0)E(0,0)$, so $A(p,q)$ is in a D-series group if and only if $E(p,q)$ is in the group. Additionally, since $A(p,q)E(0,0)=E(p,q)$, all elements of type $E$ are obtainable by multiplying an element of type $A$ by $E(0,0)$. Since in this basis the $A$ matrices are diagonal, and $E(0,0)$ is represented by the above matrix (\ref{eq:Ematrix}), this implies the claim for elements of type $E$. A similar argument holds for the elements of type $C$.} In our previous arguments eliminating $\Delta(3n^2)$ as a possibility, we showed that none of the generators $R_i$ can be traceless, so each of the $R_i$ must be of type $A$. \ Again this is a contradiction since elements of type $A$ generate an abelian group and $G_M$ is nonabelian. \ Hence $G_M$ cannot be an irrep of one of the C-series groups.

Next we turn our attention to the D-series finite subgroups of $SU(3)$. We begin by showing that $G_M$ cannot be an irrep of $\Delta(6n^2)$ for any $n$, and we will later generalize this to eliminate all D-series groups as possbilities. \ The group $\Delta(6n^2)$ contains $6$ families of conjugacy classes, labeled by types $A,B,C,D,E,F$ and by integers $p,q$ as above. \ The three-dimensional irreps of $\Delta(6n^2)$ are again labeled by $(m_1,m_2)$, which now take values in $(m,0),(0,m)$ or $(m,m)$, as well as $t\in\{0,1\}$. \ The character of each element is
\begin{align}
\Tr(A(p,q))&= e^{\frac{2\pi i}{n}\left(m_1 p +m_2 q\right)}+e^{\frac{2\pi i}{n}\left(m_1 q -m_2(p+ q)\right)}+e^{\frac{2\pi i}{n}\left(-m_1( p+q) +m_2 p\right)} \label{eq:6n2A}\\
\Tr(B(p,q))&=(-1)^t e^{\frac{2\pi i}{n}\left(m_1 p + m_2 q\right)} \\
\Tr(D(p,q))&=(-1)^t e^{\frac{2\pi i}{n}\left(m_1(\frac{n}{2}-p-q)  +m_2 p \right)} \\
\Tr(F(p,q))&=(-1)^t e^{\frac{2\pi i}{n}\left(m_1 q +m_2(\frac{n}{2}-p-q)\right)} \\
\Tr(C(p,q))&=\Tr(E(p,q))=0 \label{eq:6n2C}
\end{align}

We now eliminate the possibility that $G_M$ is an irrep of $\Delta(6n^2)$ for any $n$. Again assume by way of contradiction that $G_M$ is an irrep of $\Delta(6n^2)$ for some $n$. \ Then each $R_i$ must be in one of the types $A,B,C,D,E,F$, and each trace $T_i$ must have the corresponding character from Eqs. (\ref{eq:6n2A})-(\ref{eq:6n2C}). \ As noted previously each $T_i$ cannot be $0$, so in fact each $R_i$ must be of type $A,B,D$ or $F$. \ Furthemore, we will show the following Lemma:

\begin{lem} If $G_M $ is an irrep of $\Delta(6n^2)$, then all $R_i$ of types $B$, $D$ or $F$ are of the same type. \label{lem:6n2}
\end{lem}

By Lemma \ref{lem:6n2}, some of the $R_i$'s belong to a single type $B$, $D$ or $F$ while the remaining $R_i$'s are of type $A$. \ However, by examining the multiplication table for this group provided in \cite[Table VIII]{Fairbairn1964}, one can see that any number of elements of type $A$ plus any number of elements from a single type $B$, $D$, or $F$ cannot generate the entire group. \ This contradicts our assumption that the $R_i$'s generate an irrep of $\Delta(6n^2)$. \ Hence $G_M$ is not an irrep of $\Delta(6n^2)$ so $G$ cannot be $\Delta(6n^2)$ by Claim \ref{clm:irrep}.

We now prove Lemma \ref{lem:6n2} before continuing the proof of Claim \ref{inf}.

\begin{proof}[Proof of Lemma \ref{lem:6n2}]

Assume that $G_M$ is an irrep of $\Delta(6n^2)$. \ We will show that all of the $R_i$'s of types $B$, $D$ or $F$ are of the same type. \ We proceed by enumerating all pairs $R_i$, $R_j$ for $i\neq j$ and show that it is not possible for both $R_i$ and $R_j$ to be of distinct types $B$, $D$ or $F$.

Let $\alpha=a+bi$ where $a$ and $b$ are real. \ If $R_i$ is of type $B,D$ or $F$, then $T_i$ has norm $1$, which imposes the following equations on $a$ and $b$:
\begin{align}
|T_1|^2=1 \Rightarrow (a^2+b^2)^2 +4[a^2(1-a)+b^2(3+a)] =1 \label{eq:t1norm1}\\
|T_2|^2=1 \Rightarrow (a^2+b^2)^2 +4[a^2(1+a)+b^2(1-3a)] =1\label{eq:t2norm1}\\
|T_3|^2=1 \Rightarrow (a^2+b^2)^2 +4b^2=1\label{eq:t3norm1}
\end{align}

First suppose that $R_1$ and $R_2$ are members of distinct types $B$, $D$, or $F$. \ Then $|T_1|=|T_2|=1$. \ The only solutions to Eqs. (\ref{eq:t1norm1}) and (\ref{eq:t2norm1}) in which $0<|\alpha|^2=a^2+b^2<1$ are $\left(a=0, b=\pm\sqrt{\sqrt{5}-2}\right)$ and $\left(a=\pm\frac{1}{2}\sqrt{3(\sqrt{5}-2)},b=\pm\frac{1}{2}\sqrt{\sqrt{5}-2}\right)$. \ Note also that the product $R_1 R_2$ must be of type $C$ or $E$ according to the group multiplication table in \cite[Table VIII]{Fairbairn1964}. \ Hence the trace of $R_1 R_2$ must be $0$ if $G_M$ is an irrep of $\Delta(6n^2)$. \ This implies that
\begin{equation}
\Tr(R_1 R_2)= \alpha^3 - {\alpha^*}^3+|\beta|^2\left(1+\beta +\beta^* -|\alpha|^2\right) - |\alpha|^2=0 \label{eq:trr1r2}
\end{equation}
Since we have $\alpha=a+bi$ where the values of $a$ and $b$ are one of the six possibilities above, one can see that there is no $\beta$ which satisfies Eq. (\ref{eq:trr1r2}). \ Indeed, note that $\alpha^3 - {\alpha^*}^3$ is nonzero and pure imaginary, while the rest of the expression is real, so the terms in Eq. (\ref{eq:trr1r2}) cannot sum to zero. \ This provides the desired contradiction. \ We conclude that $R_1$ and $R_2$ cannot be of distinct types $B$, $D$, or $F$.

Next suppose that $R_1$ and $R_3$ are of distinct types $B$, $D$ or $F$. \ Then $|T_1|=|T_3|=1$. \ If $\alpha=a+bi$ as before, Eqs. (\ref{eq:t1norm1}) and (\ref{eq:t3norm1}), combined with the fact that $0<|\alpha|^2=a^2+b^2<1$, imply that $a=0$ and $b=\pm\sqrt{\sqrt{5}-2}$. \ Again, using the group multiplication table in \cite[Table VIII]{Fairbairn1964} we must have that $R_1R_3$ is of type $C$ or $E$ so
\begin{equation}
\Tr(R_1 R_3) = \alpha^3 +\alpha^*|\alpha|^2 +{\alpha^*}^2 +|\beta|^2(1+\beta+\beta^*+\alpha^2)=0 \label{eq:trr1r3}
\end{equation}
Since $\alpha = \pm i \sqrt{\sqrt{5}-2}$, this is a contradiction---for the terms $\alpha^3 +\alpha^*|\alpha|^2$ of Eq. (\ref{eq:trr1r3}) are nonzero and pure imaginary while the remaining terms are real. \ Hence $R_1$ and $R_3$ cannot be of distinct types $B$, $D$, or $F$.

Finally suppose that $R_2$ and $R_3$ are of distinct types $B$, $D$, or $F$. \ Then $|T_2|=|T_3|=1$. \ If $\alpha=a+bi$ then the only solutions to Eqs. (\ref{eq:t2norm1}) and (\ref{eq:t3norm1}) in which $0<|\alpha|^2=a^2+b^2<1$ are $\left(a=0,b=\pm\sqrt{\sqrt{5}-2}\right)$ and $\left(a\approx 0.437668, b\approx\pm0.457975\right)$. \ Furthermore using the group multiplication table in \cite[Table VIII]{Fairbairn1964} we must have that $R_2R_3$ is of type $C$ or $E$ so
\begin{equation}
\Tr(R_2 R_3) = \alpha^2 -{\alpha^*}^3 -\alpha |\alpha|^2 +|\beta|^2(\alpha \beta^* - \alpha^* \beta^* -2\alpha^*)=0 \label{eq:trr2r3}
\end{equation}
With slightly more work, one can again check that Eq. (\ref{eq:trr2r3}) cannot be satisfied with the above values of $\alpha$, under the additional constraint that $|\alpha|^2+|\beta|^2=1$, providing the desired contradiction.  Hence $R_2$ and $R_3$ cannot be of distinct types $B$, $D$, or $F$, which completes the proof of Lemma \ref{lem:6n2}.

\end{proof}

We have therefore eliminated the possiblity that $G_M$ is an irrep of $\Delta(6n^2)$ for any $n$, and so $G \neq \Delta(6n^2)$ by Claim \ref{clm:irrep}.

We now extend this argument to eliminate the possibility that $G_M$ is an irrep of any of the D-series groups. \ In Appendix G of \cite{GrimusLudl2012}, Grimus and Ludl show that for any three-dimensional irrep of a D-series group, there exists a basis (and an ordering of that basis) in which all elements of the $A$ conjugacy classes are represented by diagonal matrices, and $E(0,0)$ and $B(0,0)$ are represented by  
\begin{equation} \label{eq:EBmatricesDgroups}
\begin{matrix}
E(0,0) \rightarrow \left(\begin{matrix}0 & 1 & 0 \\ 0 & 0 &1 \\ 1 & 0 & 0 \end{matrix}\right) & B(0,0) \rightarrow \pm \left(\begin{matrix}1 & 0 & 0 \\ 0 & 0 &1 \\ 0 & 1 & 0 \end{matrix}\right).
\end{matrix}
\end{equation}
From this it can be easily shown that in any three-dimensional irrep of a D-series group, all elements of types $C$ and $E$ are represented by traceless matrices, and all elements of types $B$, $D$ and $F$ are represented by matrices whose trace has unit norm.\footnote{The fact that matrices representing elements of types $C$ and $E$ are traceless follows from the previous arguments regarding C-series groups. The fact that matrices representing elements of types $B$, $D$ and $F$ have traces of norm 1 follows by an identical argument since $A(p,q)=B(p,q)B(0,0)$, $A(p,q)B(0,0)=B(p,q)$, $A$ matrices have diagonal representatives, and $B(0,0)$ is represented by the above in this basis. A similar argument holds for elements of types $D$ and $F$.} In our previous arguments eliminating $\Delta(6n^2)$ as a possibility, we showed none of our generators $R_i$ can be traceless, and those of trace norm 1 are of the same type. Hence some of our generators are of type $A$ and the remainder are of a single type $B$, $D$ or $F$. Again this is a contradiction since any number of elements of type $A$ and any number of elements of a single type $B$, $D$ or $F$ do not suffice to generate any D-series group - in particular they cannot generate $E(0,0)$. Hence $G_M$ cannot be an irrep of any of the D-series groups. This concludes the proof that $G_M$ cannot be an irrep of any of the “dihedral-like” subgroups of $SU(3)$.

Finally we will show that $G$ is not a finite subgroup of $SU(2)$. \ Since $SU(2)$ is a double cover of $SO(3)$, if $G$ is a finite subgroup of $SU(2)$, then $G$ must be either a finite subgroup of $SO(3)$ or else the double cover of such a subgroup.
\ We first eliminate the finite subgroups of $SO(3)$. \ The dihedral and cyclic subgroups have no three-dimensional irreps; hence $G$ cannot be one of these by Claim \ref{clm:irrep}. \ The icosahedral subgroup is isomorphic to $\Sigma(60)$ so has already been eliminated. \ The octahedral and tetrahedral subgroups do have three-dimensional irreps. \ However, the characters of their elements all lie in the set $\{0,\pm1,\pm3\}$, so these can be eliminated just as the exceptional groups of $SU(3)$ were eliminated.

Now all that remains are double covers of the finite subgroups of $SO(3)$. \ The binary dihedral groups, also known as the dicyclic groups, have no three-dimensional irreps, so $G$ cannot be a binary dihedral group by Claim \ref{clm:irrep}. \ The binary tetrahedral group has one three-dimensional irrep, with character values in the set $\{0,\pm1,\pm3\}$. \ So $T_3$ cannot be in this set as noted above.

The binary octahedral group has two three-dimensional irreps, with character values also in $\{0,\pm1,\pm3\}$, so is likewise eliminated. \ The binary icosahedral group has two three-dimensional irreps, with all characters in the set $\{0,-1,3,\frac{\sqrt{5}\pm1}{2}\}$. \ As discussed in the case of $\Sigma(60)$, our traces cannot take these values.

In summary, by enumeration of the finite subgroups of $SU(3)$, we have shown that $G$ cannot be finite.
\end{proof}

\begin{cor} \label{cor:GLie} $G$ is a continuous (Lie) subgroup of $SU(3)$.
\end{cor}
\begin{proof}
$G$ is infinite by Claim \ref{inf}. \ Furthermore $G$ is closed because it is the set of matrices \emph{densely} generated by $\{R_1,R_2,R_3\}$. \ It is well-known that a closed, infinite subgroup of a Lie group is also a Lie group (this is Cartan's theorem \cite{Cartan1952}). \ The corollary follows.
\end{proof}

Next we show that $G$ must be either $SO(3)$, $SU(2)$ or $SU(3)$. \
Furthermore, the set of matrices $G_M$ densely generated by $\{R_1,R_2,R_3\}$
consists of either all $SO(3)$ matrices or all $SU(3)$ matrices.

\begin{claim}
$G$ is either $SO(3)$, $SU(2)$, or $SU(3)$. \ Furthermore, $G_M$ consists of either all $3 \times 3$ special unitary matrices (if the beamsplitter $b$ has a non-real entry), or all $3 \times 3$ special orthogonal matrices (if $b$ is real). \label{GMso3}
\end{claim}
\begin{proof}

Since $R_1$, $R_2$, and $R_3$ do not commute, $G$ is nonabelian. \ By Corollary \ref{cor:GLie}, we know $G$ is a Lie group, and furthermore $G$ is closed. \ The nonabelian closed connected Lie subgroups of $SU(3)$ are well-known \cite{Brocker2003}: they are $SU(3)$, $SU(2)\times U(1)$, $SU(2)$, and $SO(3)$. \ Meanwhile, the closed disconnected Lie subgroups of $SU(3)$ are $\Delta(3\infty)$ and $\Delta(6\infty)$, as described in \cite{Fairbairn1964}.

Note that $\Delta(3\infty)$ and $\Delta(6\infty)$ are the analogues of $\Delta(3n^2)$ and $\Delta(6n^2)$ as $n\rightarrow \infty$. \ Our above arguments showing that $G\neq \Delta(3n^2)$ and $G\neq \Delta(6n^2)$ carry over in this limit, because at no point did we use the fact that $n$ or $m$ were finite. \ Therefore $G$ cannot be either of these continuous groups.

By Claim \ref{clm:irrep}, $G$ has a three-dimensional irrep. \ Of the remaining groups, only $SU(2)$, $SO(3)$, and $SU(3)$ have three-dimensional irreps. \ Furthermore, it is well known that the only three-dimensional irrep of $SU(2)$ is as $SO(3)$. \ This is because $SU(2)$ has exactly one irrep in each finite dimension (See \cite[Section II.5]{Brocker2003} or \cite{Yuan} for details), and $SU(2)$ has an obvious representation as $SO(3)$ via the fact that $SU(2)$ is a double cover of $SO(3)$. \ Since we are only concerned with the set of matrices $G_M$ generated, without loss of generality we can assume $G$ is either $SO(3)$ or $SU(3)$.

It is well-known that the only three-dimensional irrep of $SU(3)$ is the natural one, as the group of all $3\times 3$ special unitary matrices (\cite[Section VI.5]{Brocker2003}). \ Likewise, the only three-dimensional irrep of $SO(3)$ is the natural one, up to conjugation by a unitary \cite{Brocker2003}. \ Hence $G_M$ consists of either all $3 \times 3$ special unitary matrices (case A), or all $3 \times 3$ special orthogonal matrices conjugated by some unitary $U$ (case B).

We now show that if the beamsplitter $b$ is real, then we are in case B and without loss of generality the conjugating unitary $U$ is real. \ Hence $G_M$ is the set of all $3\times 3$ orthogonal matrices. \ Otherwise, if $b$ has a complex entry, we will show we are in case A and $G_M$ is the set of all $3\times 3$ special unitary matrices.

First, suppose $b$ is real. \ Then all matrices in our generating set are orthogonal, so all matrices in $G_M$ are orthogonal. \ Hence we are in case $B$, and since all matrices in $G_M$ are real, without loss of generality $U$ is a real matrix as well.

Now suppose that $b$ has a complex entry. \ Then either $\alpha$ or $\beta$ are not real. \ First, suppose $\alpha$ is not real. \ Then $\mathrm{Tr}(R_1) = \alpha^2 - 2\alpha^*$ is not real because $0<|\alpha|<1$. \ But since conjugating a matrix by a unitary preserves its trace, and we are in case $B$, the traces of all matrices in $G_M$ must be real. \ In particular $\mathrm{Tr}(R_1)$ must be real, which is a contradiction. \ Therefore if $\alpha$ is not real then we must be in case A.

Next, suppose $\beta$ is not real. \ Then we can obtain a similar contradiction. \ Let $\beta = p+qi$ where $p$ and $q$ are real. \ By direct calculation one can show that $\mathrm{Im}\left(\mathrm{Tr}(R_1 R_2 R_3 R_1)\right) = |\beta|^4\left(\beta^{*2} +2\beta \right)$. \ Since our beamsplitter is nontrivial, $|\beta|^4 \neq 0$, so this quantity is $0$ if and only if $\beta^{*2}+2\beta=0 \Leftrightarrow 2q(1-p)=0$. \ But this cannot occur, since $q\neq0$ (because $\beta$ is not real), and $1-p\neq 0$ (because the beamsplitter is nontrivial). \ Hence in this case $\mathrm{Tr}(R_1 R_2 R_3 R_1)$ is imaginary, which contradicts the fact we are in case B. \ Therefore if $\beta$ is not real then we must be in case A, which completes the proof.
\end{proof}

Theorem \ref{o3_generation} follows from Claim \ref{GMso3}. \ Having proved
our main result, we can now easily show two alternative versions of the
theorem as well.

\begin{cor}
Any nontrivial two-mode optical gate $g=\bigl(\begin{smallmatrix}a & b \\ c & d\end{smallmatrix}\bigr)$ (not necessarily of determinant $-1$), plus the set of all phaseshifters densely generates $SU(m)$ on $m\geq 3$ modes.
\end{cor}
\begin{proof}
Since $g$ is unitary we have $\det(g)=e^{i\theta}$ for some $\theta$. \ By composing $g$ with a phase of $e^{i\frac{\pi-\theta}{2}}$, we obtain a nontrivial beamsplitter $g'$ of determinant $-1$. \ The gate $g'$ is universal by Theorem \ref{o3_generation}, hence this gate set is universal as well.
\end{proof}

\begin{cor}
Any nontrivial two-mode real optical gate $g$ is universal for quantum linear optics.
\end{cor}
\begin{proof}
Since $g$ is real, $g$ must have determinant $\pm1$. \ The case of $\det(g)=-1$ is handled by Theorem \ref{o3_generation}, so we now prove the $\det(g)=+1$ case. \ In this case $g$ is a rotation by an angle $\theta$. \ The fact that $g$ is nontrivial means $\theta$ is not a multiple of $\pi/2$. \ The beamsplitter actions $b_{12},b_{23},b_{13}$ can be viewed as three-dimensional rotations by angle $\theta$ about the $x$, $y$ and $z$ axes. \ So the question reduces to ``For which angles $\theta$ (other than multiples of $\pi/2$) do rotations by $\theta$ about the $x$, $y$ and $z$ axes fail to densely generate all possible rotations?"

This question is easily answered using the well-known classification of closed subgroups of $SO(3)$. \ The finite subgroups of $SO(3)$ are the cyclic, dihedral, tetrahedral, octahedral, and icosahedral groups. \  One can easily check that our gate $g$ cannot generate a representation of one of these groups, and hence densely generates some infinite group $G$. \ By the same reasoning as in Corollary \ref{cor:GLie}, we conclude that $G$ is a Lie subgroup of $SO(3)$.

The Lie subgroups of $SO(3)$ are $SO(3)$, $U(1)$ (all rotations about one axis) and $U(1)\times \mathbb{Z}_2$ (all rotations about one axis, plus a rotation by $\pi$ perpendicular to the axis). \ Again one can easily eliminate the possibility that $G$ is $U(1)$ or $U(1)\times \mathbb{Z}_2$, and hence $G$ must be all of $SO(3)$.

We have proven universality on three modes for real nontrivial $g$ with determinant $+1$. \ Universality on $m\geq 3$ modes follows by a real analogue of Reck \emph{et al}.\ \cite{Reck1994}, namely that any rotation matrix in $SO(m)$ can be expressed as the product of $O(m^2)$ real $2 \times 2$ optical gates of determinant $1$.
\end{proof}

\section{Open Questions}

\label{openqs} At the moment our dichotomy theorem only holds for
beamsplitters, which act on two modes at a time and have determinant $-1$. \
As we said before, we leave open whether the dichotomy can be extended to two%
-mode gates with determinant other than $-1$. \ Although the phases of gates
are irrelevant in the qubit model, the phases unfortunately \emph{are}
relevant in linear optics---and that is the source of the difficulty. \ Note
that the previous universality result of Reck \emph{et al}.\ \cite{Reck1994} simply
assumed that arbitrary phaseshifters were available for free, so this issue
did not arise.

Another open problem is whether our dichotomy can be extended to $k$-mode
optical gates for all constants $k$. \ Such a result would complete the
linear-optical analogue of the dichotomy conjecture for standard quantum
circuits. \ The case $k=3$ seems doable because the representations of all exceptional
finite subgroups of $SU(4)$ are known \cite{Hanany2001}. \ But already the
case $k=4$ seems more difficult, because the representations of all finite
subgroups of $SU(5)$ have not yet been classified. \ Thus, a proof for
arbitrary $k$ would probably require more advanced techniques in
representation theory.

\section{Acknowledgements}

We thank Bela Bauer for pointing out Ref. \cite{Ludl2011} and Patrick Otto Ludl for helpful correspondence. A.B. was supported by the National Science Foundation Graduate Research
Fellowship under Grant No.\ 1122374 and by the Center for Science of
Information (CSoI), an NSF Science and Technology Center, under grant
agreement CCF-0939370. \ S.A. was supported by the National Science
Foundation under Grant No.\ 0844626, by a TIBCO Chair, and
by an Alan T.\ Waterman Award.

\bibliographystyle{abbrv}
\bibliography{allpapers}

\begin{thebibliography}{10}

\bibitem{Aaronson2010}
S.~Aaronson and A.~Arkhipov.
\newblock {The Computational Complexity of Linear Optics}.
\newblock {\em Theory of Computing}, 9(4):143--252, 2013.

\bibitem{Bartlett2003}
S.~D. Bartlett and B.~C. Sanders.
\newblock {Requirement for quantum computation}.
\newblock {\em Journal of Modern Optics}, 50:2331--2340, 2003.

\bibitem{BJS}
M.~Bremner, R.~Jozsa, and D.~Shepherd.
\newblock {Classical simulation of commuting quantum computations implies
  collapse of the polynomial hierarchy}.
\newblock {\em Proc. Roy. Soc. London}, A467(2126):459--472, 2010.

\bibitem{Brocker2003}
T.~Br\"{o}cker and T.~tom Dieck.
\newblock {\em {Representations of Compact Lie Groups}}.
\newblock Springer, 2003.

\bibitem{Cartan1952}
\'{E}.~Cartan.
\newblock {La th\'{e}orie des groupes finis et continus et l'analysis situs}.
\newblock {\em M\'{e}morial des Sciences Math\'{e}matiques}, 42:1--61, 1952.

\bibitem{Cerf1998}
N.~J. Cerf, C.~Adami, and P.~G. Kwiat.
\newblock Optical simulation of quantum logic.
\newblock {\em Phys. Rev. A}, 57:R1477--R1480, 1998.

\bibitem{Dawson2005}
C.~Dawson and M.~Nielsen.
\newblock {The Solovay-Kitaev Algorithm}.
\newblock {\em Quantum Information and Computation}, 6(1):81--95, 2006.

\bibitem{Fairbairn1964}
W.~M. Fairbairn, T.~Fulton, and W.~H. Klink.
\newblock {Finite and Disconnected Subgroups of SU(3) and their Application to
  the Elementary-Particle Spectrum}.
\newblock {\em Journal of Mathematical Physics}, 5(8):1038, 1964.

\bibitem{Gottesman2008}
D.~Gottesman.
\newblock {The Heisenberg representation of quantum computers}.
\newblock In {S. P. Corney, R. Delbourgo} and P.~D. Jarvis, editors, {\em
  Group22: Proceedings of the XXII International Colloquium on Group
  Theoretical Methods in Physics}, volume~1, pages 32--43, Cambridge, MA, 1999.
  International Press.

\bibitem{GottesmanKitaevPreskill}
D.~Gottesman, A.~Kitaev, and J.~Preskill.
\newblock Encoding a qubit in an oscillator.
\newblock {\em Phys. Rev. A}, 64:012310, 2001.


\bibitem{GrimusLudl2012}
W. ~Grimus and P. O. ~Ludl.
\newblock Finite flavour groups of fermions.
\newblock {\em J. Phys. A: Math. Theor.} 45:233001, 2012.

\bibitem{GrimusLudl2014}
W. ~Grimus and P. O. ~Ludl.
\newblock On the characterization of the SU(3)-subgroups of type C and D.
\newblock {\em J. Phys. A: Math. Theor.} 45:233001, 2012.



\bibitem{Hanany1999}
A.~Hanany and Y.~He.
\newblock {Non-abelian finite gauge theories}.
\newblock {\em Journal of High Energy Physics}, 1999(02):013, 1999.

\bibitem{Hanany2001}
A.~Hanany and Y.-H. He.
\newblock {A monograph on the classification of the discrete subgroups of
  SU(4)}.
\newblock {\em Journal of High Energy Physics}, 2001(02):027, 2001.

\bibitem{Jordan}
S.~P. Jordan.
\newblock {Permutational quantum computing}.
\newblock {\em Quantum Information and Computation}, 10(5):470--497, 2010.

\bibitem{Jozsa2013}
R.~Jozsa and M.~Nest.
\newblock {Classical simulation complexity of extended Clifford circuits}.
\newblock {\em Quantum Information and Computation}, 14(7):633-648, 2014.




\bibitem{Knill1998}
E.~Knill and R.~Laflamme.
\newblock {Power of One Bit of Quantum Information}.
\newblock {\em Physical Review Letters}, 81(25):5672--5675, 1998.

\bibitem{Knill2001}
E.~Knill, R.~Laflamme, and G.~J. Milburn.
\newblock {A scheme for efficient quantum computation with linear optics.}
\newblock {\em Nature}, 409(6816):46--52, 2001.

\bibitem{Lloyd1995}
S.~Lloyd.
\newblock {Almost any quantum logic gate is universal}.
\newblock {\em Physical Review Letters}, 75(2):346--349, 1995.


\bibitem{Ludl2011}
P. O. ~Ludl.
\newblock Comments on the classification of the finite subgroups of SU(3).
\newblock {\em J. Phys. A: Math. Theor.} 44:255204, 2011.



\bibitem{Nielson2000}
M.~Nielsen and I.~Chuang.
\newblock {\em {Quantum Computation and Quantum Information}}.
\newblock Cambridge University Press, 2000.

\bibitem{Reck1994}
M.~Reck, A.~Zeilinger, H.~J. Bernstein, and P.~Bertani.
\newblock {Experimental realization of any discrete unitary operator}.
\newblock {\em Physical Review Letters}, 73(1):58--61, 1994.

\bibitem{Shi2002}
Y.~Shi.
\newblock {Both Toffoli and controlled-NOT need little help to do universal
  quantum computation}.
\newblock {\em Quantum Information and Computation}, 3(1):84--92, 2003.

\bibitem{Yuan}
Q.~Yuan.
\newblock
  http://qchu.wordpress.com/2011/06/26/the-representation-theory-of-su2/, 2011.

\end{thebibliography}

\end{document}